\documentclass[12pt]{article}
\usepackage[margin=1.0in]{geometry}

\usepackage{graphicx}
\usepackage{amsthm,amsmath,amssymb, amsfonts}
\usepackage{color}
\usepackage[dvipsnames]{xcolor}

\definecolor{myblue}{RGB}{94, 129, 181}
\definecolor{myorange}{RGB}{225, 156, 36}
\definecolor{mygreen}{RGB}{143, 176, 50}

\usepackage{array}

\usepackage{url} 

\usepackage{setspace}
\usepackage{caption}
\def\R{\mathbb{R}}

\def\sF{{\mathcal F}}

\def\sA{{\mathcal A}}

\def\sM{{\mathcal M}}

\def\E{\mathbb{E}}

\def\sF{\mathcal{F}}
\def\P{\mathbb{P}}

\def\bQ{\mathbb{Q}}

\def\N{\mathbb N}

\newcommand{\ot}{\bar{t}}

\numberwithin{equation}{section}
\theoremstyle{plain}                
\newtheorem{theorem}{Theorem}[section]
\newtheorem{lemma}[theorem]{Lemma}
\newtheorem{proposition}[theorem]{Proposition}

\theoremstyle{definition}           
\newtheorem{definition}[theorem]{Definition}

\theoremstyle{remark}               

\begin{document}

\pagenumbering{arabic} \pagestyle{plain}
\begin{center}
  {\large \bf Endogenous noise trackers in a Radner equilibrium}
  \ \\ \ \\
Jin Hyuk Choi\footnote{Jin Hyuk Choi is supported by the National Research Foundation of Korea (NRF) grant funded by the Korea government (MSIT) (No. 2020R1C1C1A01014142 and No. 2021R1A4A1032924).} \\
Ulsan National Institute of Science and Technology (UNIST)
\ \\ \ \\
Kim Weston\footnote{
Kim Weston acknowledges support by the National Science Foundation under Grant No. DMS\#1908255 (2019-2022). Any opinions, findings and conclusions or recommendations expressed in this material are those of the author and do not necessarily reflect the views of the National Science Foundation (NSF).}\\
Rutgers University
\ \\

\ \\

\today
\end{center}
\vskip .5in

\abstract{We prove the existence of an incomplete Radner equilibrium in a model with exponential investors and an endogenous noise tracker.  We analyze a coupled system of ODEs and reduce it to a system of two coupled ODEs in order to establish equilibrium existence.  As an application, we study the impact of the endogenous noise tracker on welfare by comparing to a model with an exogenous noise trader.  We show that the aggregate welfare in the endogenous noise tracker model is bigger for a sufficiently large stock supply, but the welfare comparison depends in a non-trivial manner on the other model parameters.


\section{Introduction}\label{section:intro}

We study a financial equilibrium in an economy with a strategic noise trader, which we call an endogenous noise tracker.  We prove global existence of a Radner equilibrium and explore its equilibrium implications.  In the model, a finite number of utility-maximizing exponential investors and a single endogenous noise tracker trade in a continuous-time financial market.   The exponential investors derive utility from consumption of terminal wealth in a pure-exchange economy{\color{black}, while the endogenous noise tracker is incentivized to track a noisy target and maximizes expected wealth.}

The use of noise traders is a standard modeling tool in equilibrium theory.  Price-inelastic noise traders often model uninformed traders' demands or serve as a way to express unmodeled demands.  These traders are introduced in a wide array of works such as Grossman and Stiglitz~\cite{GS80AER}, Kyle~\cite{K85E}, Hellwig~\cite{H80JET}, Vayanos~\cite{V01JF}, and the vast literature that has followed.  These exogenously-determined, price-inelastic noise traders do not choose demand processes based on a decision problem.  In contrast, Sannikov and Skrzypacz~\cite{SS16wp} and Choi et.\,al.~\cite{CLS21MAFE} use noise traders who are incentivized to track a noisy target.  Such noise trackers maximize their expected wealth from trading but are penalized for deviating from a Brownian motion.  Our equilibrium incorporates noise trading of the second kind, where the endogenous noise tracker is incentivized to track a noisy signal.

Our main contribution is to prove the existence of a Radner equilibrium.  Incomplete Radner equilibria are difficult to study mathematically because standard tools and simplifications, such as a representative agent, are often not applicable.  Approaches with partial differential equations (PDEs) and backward stochastic differential equations (BSDEs) have been successful in proving abstract existence; see, for example, Choi and Larsen~\cite{CL15FS}, Xing and \v{Z}itkovi\'{c}~\cite{XZ18AP}, and Escauriaza et.\,al.~\cite{ESX20wp}.  Exogenous noise in equilibrium has been successfully studied in G\^{a}rleanu and Pedersen~\cite{GP16JET} and Bouchard et.\,al.~\cite{BFHMK18FS} by relaxing the utility maximizers' problems from exponential investors to linear-quadratic optimizers.

Our approach derives a system of coupled ordinary differential equations (ODEs) that are functions of time.  This derivation is possible because we propose a functional form for the equilibrium stock price, which is similar to the Nash equilibrium stock price structure in Chen et.\,al.~\cite{CCLS20wp}.  The ODE system consists of 15 coupled equations.  Upon inspection, we reduce the system to two core coupled ODEs, which allows us to prove the existence of global solutions.

Given the prevalence of noise traders of different varieties in the equilibrium literature, it is natural to consider the equilibrium implications of such modeling choices.  We ask the question:
\begin{center}
  \it Which modeling scenario do the utility-maximizing investors prefer -- an exogenous noise trader or an endogenous noise tracker?
\end{center}
Our analysis shows that for a large enough stock supply, aggregate welfare is larger in the endogenous noise tracker model compared to an exogenous noise trader model. However, we observe non-trivial and non-monotonic dependencies of the aggregate welfare difference in the other model parameters and across multiple parameters. Therefore, for a specific choice of the model parameters, numerical computations are necessary to decide which type of noise trader can have a positive/negative impact on the aggregate welfare.

The paper is organized as follows.  Section~\ref{section:set-up} describes the model setting.  Section~\ref{section:main} states and proves the main result, Theorem~\ref{thm:radner-endogenous}.  The economic implications of an endogenous noise tracker are explored by a welfare analysis in Section~\ref{section:welfare}.  We compare the aggregate welfare in our model against a baseline model, which is a Radner equilibrium model with an exogenous noise trader.  The exogenous noise trader model is included in the Appendix, Section~\ref{section:exogenous}.

\section{Model set-up}\label{section:set-up}
For simplicity, we normalize the trading time horizon as $1$. Let $\P$ be a probability measure, and $(D_t, Y'_t)_{t\in[0,1]}$ be two independent, one-dimensional Brownian motions under $\P$ with constant values $(D_0, Y_0')$ at time $t=0$, zero drifts, and constant volatilities $(\sigma_D,\sigma_{Y'})$. The augmented standard Brownian filtration is denoted by
\begin{align}\label{sFi}
\sF_{t}:=\sigma(D_u,Y'_u)_{u\in[0,t]},\quad t\in[0,1].
\end{align}

The market consists of two traded securities: a bank account and stock.  The bank account is in zero-net supply with a constant zero interest rate.  The stock is in a constant net supply of $\Sigma$, where $\Sigma\geq 0$ is a nonnegative constant.  All prices are denominated in units of a single consumption good.

The equilibrium stock price will be determined endogenously in equilibrium as a continuous semimartingale and is denoted $S = (S_t)_{t\in[0,1]}$.  The terminal dividend is modeled by terminal value $D_1$ of the Brownian motion $D$.  The terminal equilibrium stock price is exogenously pinned down as:
\begin{align}\label{pp10}
S_1 =D_1,\quad \P\text{-a.s.}
\end{align}

We model the stock positions of a group of $j\in \{1,...,I\}$, $I\in\N$, utility-maximizing investors.  Their stock position processes over time are denoted $\theta_j=\left(\theta_{j,t}\right)_{t\in[0,1]}$, and investor $j$ is endowed with an initial stock position of $\theta_{j,0-}\in\R$ and a zero initial position in the bank account.  

There are a range of possibilities when defining admissible trading strategies for exponential investors, such as in Delbaen et.\,al.~\cite{6AP} and Biagini and Sirbu~\cite{BS12S}.  Here, we employ the approach of Choi and Larsen~\cite{CL15FS}, who also study an equilibrium with exponential investors.  We let $\sM$ denote the collection of $\P$-equivalent probability measures under which the stock price process $S$ is a local martingale.  A strategy $\theta$ is called \textit{admissible for $\bQ\in\sM$} if $\theta$ is adapted to $\left(\sF_t\right)_{t\in[0,1]}$, measurable, $S$-integrable, and the wealth process $X^\theta$ is a $\bQ$-supermartingale, where $X^\theta$ is define below in~\eqref{def:wealth}.  We denote the collection of admissible strategies with the associated measure $\bQ$ by $\sA(\bQ)$.

The wealth at time $t$ associated with $\theta\in\sA(\bQ)$ is denoted by $X^\theta_t$ where
\begin{equation}\label{def:wealth}
  X^\theta_t := \theta_{0-}S_0 + \int_0^t \theta_u\,dS_u,\quad t\in[0,T].
\end{equation}

The utility-maximizing investors seek to maximize their expected utility from terminal wealth.  We assume that all investors have identical exponential utility functions
\begin{align}\label{eqn:exp-util}
 -\exp\left(-a x \right)\; \text{ at time }t=1,\quad x\in\R,
\end{align}
where $a>0$ is their common risk aversion coefficient. Based on the utility function in \eqref{eqn:exp-util}, investor $j$ seeks to solve
\begin{align}\label{optproblem}
\sup_{ \theta\in \sA(\bQ)}\E\left[-\exp\left(-a X^\theta_{1}\right)\right].
\end{align}
The measure $\bQ$ is not investor-specific, unlike in Choi and Larsen~\cite{CL15FS}, because our investors all share a common risk aversion coefficient.  It is possible to extend our model to include heterogenous risk aversion coefficients.  Because, as we shall see below, our model with identical risk aversion coefficients already produces ambiguous welfare implications, we do not pursue such an extension.

We introduce an endogenous noise tracker, similar to Sannikov and Skrzypacz~\cite{SS16wp}.  Rather than hold an exogenously determined number of shares, the endogenous noise tracker is incentivized to track a noisy signal through her decision problem. The exogenously given target of the noise tracker is denoted by 
\begin{align}\label{pp1}
Y_t :=Y_0 +\int_0^t Y'_udu,\quad t\in[0,1],
\end{align}
where we recall that $Y'$ is a Brownian motion with zero drift and constant volatility $\sigma_{Y'}$. The noise tracker has linear-quadratic preferences that encourage her to trade towards the target $Y$ in \eqref{pp1}:
\begin{align}\label{optproblem_tracker}
\sup_{ \theta\in \sA_N}\E\left[ X_{1}^\theta-\int_0^1 \kappa\left(\theta_{t}-Y_t\right)^2dt\right],
\end{align}
where the constant $\kappa>0$ measures the motivation to track the noise $Y$, and $\sA_N$ is the collection of strategies admissible for the noise tracker. 
The noise tracker's trading admissibility condition differs from the exponential utility maximizers' condition since her optimization problem is linear-quadratic rather than exponential.  A trading strategy $\theta$ is \textit{admissible for the noise tracker} if $\theta$ is adapted to $(\sF_t)_{t\in[0,1]}$, measurable, and $\E\left[\int_0^1 \theta_u^2du \right]<\infty$.  

We denote the noise trackers's stock holdings over time by $\theta_{N,t}$.  She holds $\theta_{N,0-}=Y_0$ shares initially, and, like the utility maximizers, holds a zero position initially in the bank account.  With the presence of the endogenous noise tracker, the stock market clearing condition becomes
\begin{equation}\label{eqn:clearing2}
  \sum_{j=1}^I\theta_{j,t}+\theta_{N,t} = \Sigma, \quad t\in[0,1].
\end{equation}

The exogenous noise in our model is similar to G\^arleanu and Pedersen~\cite{GP16JET} and Bouchard et.~al.~\cite{BFHMK18FS}; however, our approach to establishing equilibrium existence is different because we use exponential investors instead of their linear-quadratic objectives.  The earlier works established existence using pointwise optimization, whereas our existence result relies on analysis of ODEs derived from the investors' Hamilton-Jacobi-Bellman PDEs.

In order to prove the existence of a Radner equilibrium, we conjecture a form for the exponential investors' value functions and derive a coupled system of ODEs.  Using time $t$, wealth $x$, and states $Y$ and $Y'$ as state variables for the value function, we conjecture that each investor's value function will take the form
\begin{align}
\begin{split}\label{eqn:val-fn-form}
  &V(t,x,Y,Y') \\
  &= -\exp\left(-a\left(x+g_1(t)+g_2(t)Y+g_3(t)Y'+ g_{22}(t)Y^2+ g_{23}(t)Y Y'+ g_{33}(t)Y'^2\right)\right),
  \end{split}
\end{align}
where  $g_1$, $g_2$, $g_3$, $g_{22}$, $g_{23}$, $g_{33}:[0,1]\rightarrow\R$ are smooth functions of time with terminal conditions
\begin{equation}\label{eqn:term-cond}
  g_1(1) = g_2(1) = g_3(1) = g_{22}(1) = g_{23}(1) = g_{33}(1) = 0.
\end{equation}

\begin{definition}[Radner equilibrium with endogenous noise tracking]\label{def:end-eq}
  Trading strategies $\hat\theta_1,\ldots,\hat\theta_I$ and $\hat\theta_N\in\sA_N$ and a continuous semimartingale $S=(S_t)_{t\in[0,1]}$ form a \textit{Radner equilibrium with endogenous noise tracking} if there exists a measure $\widehat\bQ$ under which $S$ is a local martingale such that
  \begin{enumerate}
    \item \textit{Strategies are optimal:}  For $j=1,\ldots,I$, we have $\hat\theta_j\in\sA(\widehat\bQ)$ solves \eqref{optproblem} with measure $\widehat\bQ$, and $\hat\theta_N\in\sA_N$ solves \eqref{optproblem_tracker}, where $S$ is the corresponding stock price process.
    \item \textit{Markets clear:} We have
      \begin{align}\label{eqn:clearing2}
 \sum_{j=1}^{I} \hat\theta_{j,t}+\hat\theta_{N,t}=\Sigma,\quad t\in[0,1].
 \end{align} 
  \end{enumerate}
\end{definition}

\section{Main result}\label{section:main}

Theorem~\ref{thm:radner-endogenous} below establishes the existence of a Radner equilibrium with endogenous noise tracking and makes the conjectured form of \eqref{eqn:val-fn-form}-\eqref{eqn:term-cond} rigorous.  

\begin{theorem}[Radner existence with endogenous noise tracking] \label{thm:radner-endogenous}\label{thm:radner-endogenous} Let $\Sigma\ge0$, $a, \sigma_D^2, \kappa>0$, and $\sum_{j=1}^I \theta_{j,0-} + Y_0=\Sigma$. Then, there exists a unique smooth solution to the coupled system of ODEs 
for $t\in[0,1]$:
\begin{align*}
\begin{split}
  g_{33}'(t) &= \frac{2 a \sigma_{Y'}^2 g_{33}(t)^2 \left(a \sigma_{Y'}^2 \beta (t)^2 \left(a \sigma_D^2+4 I \kappa\right)+\left(a \sigma_D^2+2 I \kappa\right)^2\right)}{\left(a \sigma_D^2+a \sigma_{Y'}^2 \beta (t)^2+2 I \kappa\right)^2}-\frac{2\kappa \beta(t)}{2I\kappa+a\sigma_D^2},\quad g_{33}(1)=0,\\
  \beta'(t)&=\frac{4 a I \sigma_{Y'}^2 g_{33}(t) \beta (t) \kappa}{a \sigma_D^2+a \sigma_{Y'}^2 \beta (t)^2+2 I \kappa} - \frac{2a \kappa \sigma_D^2}{2I \kappa + a \sigma_D^2}(1-t),\quad \beta(1)=0,
\end{split}
\end{align*}
such that a Radner equilibrium with endogenous noise tracking exists. {\color{black}The functions $g_{33}$ and $\beta$ are strictly positive on $[0,1)$.}  The equilibrium stock price process is given by
\begin{align}\label{R11}
 S_t &:=D_t +\mu(t)+\alpha(t) Y_t+\beta(t) Y'_t,\quad t\in[0,1],
\end{align}
where for $t\in[0,1]$, we have
\begin{align*}
  \begin{split}
    \mu(t) &:= -\tfrac{2a \kappa \sigma_D^2 \Sigma }{2I \kappa + a \sigma_D^2}(1-t),\\
    \alpha(t) &:= \tfrac{2a \kappa \sigma_D^2}{2I \kappa + a \sigma_D^2}(1-t).
  \end{split}
\end{align*}

There exists $\widehat\bQ\in\sM$ such that the investors' optimal stock holdings $\hat\theta_j\in\sA(\widehat\bQ)$ are given by
\begin{equation}\label{def:hat-theta}
  \hat\theta_{j,t}=\frac{2 \kappa (\Sigma -Y_t)}{2I \kappa +a\sigma_D^2} - \frac{2a \sigma_{Y'}^2 \beta (t) g_{33}(t)Y'_t}{a \sigma_D^2+a \sigma_{Y'}^2 \beta (t)^2+2 I \kappa},\quad t\in[0,1],\ j\in\{1,\ldots,I\},
\end{equation}
and the noise tracker optimally holds $\hat\theta_N\in\sA_N$ where
\begin{equation}\label{def:tracker-hat-theta}
  \hat\theta_{N,t} = \frac{2 I\kappa Y_t + a\sigma_D^2 \Sigma}{2I \kappa +a\sigma_D^2} + \frac{2a I \sigma_{Y'}^2 \beta (t) g_{33}(t)Y'_t}{a \sigma_D^2+a \sigma_{Y'}^2 \beta (t)^2+2 I \kappa},\quad t\in[0,1].
\end{equation}
In the class of equilibria that has equilibrium stock prices of the form~\eqref{R11} for continuously differentiable functions $\mu$, $\alpha$, $\beta$, the equilibrium established here is unique.
\end{theorem}

{\color{black}\label{item-2-comments}
The functional form of $S$ in \eqref{R11} is inspired by the Nash equilibrium in Chen et.\,al.~\cite{CCLS20wp}.  Since the functions $\alpha$ and $\beta$ are strictly positive on $[0,1)$, the equilibrium stock price is positively correlated with the noise components $Y$ and $Y'$.  $Y'$ represents the instantaneous change in the noise levels, and the term $\alpha(t)Y'_t$ in the equilibrium stock price conveys this instantaneous impact.  $Y$ is the absolute noise level that the endogenous noise tracker is incentivized to hold.  The term $\beta(t)Y_t$ represents the persistent noise effect on the equilibrium stock price.

The function $\mu$ is negative (and zero when the supply of stock shares $\Sigma$ is zero).  The negativity of the term $\mu(t)$ in $S_t$ shows that the stock should be cheap enough to attract the risk-averse investors.

The terms $\hat\theta_j$ in \eqref{def:hat-theta} and $\hat\theta_N$ in \eqref{def:tracker-hat-theta} are the optimal stock holdings.  Since $g_{33}$ and $\beta$ are strictly positive on $[0,1)$, $\hat\theta_j$ is affected negatively by both $Y'$ and $Y$, while $\hat\theta_N$ is positively affected.  This outcome is expected because the endogenous noise tracker is incentivized to hold $Y$ shares, while the utility-maximizing investors must choose the stock holdings accordingly to clear the market.

\label{item-2-comments-kappa}
The parameter $\kappa>0$ measures the endogenous noise tracker's motivation to track the noise $Y$.  Higher values for $\kappa$ correspond to the endogenous noise tracker's increasing pressure to track $Y$ versus achieving a higher expected wealth.  In the extreme case when $\kappa\rightarrow\infty$, the endogenous noise tracker's optimization problem \eqref{optproblem_tracker} only considers the incentive to hold $Y$ shares without regard to expected wealth.  The $\kappa\rightarrow\infty$ case corresponds to an equilibrium with an exogenous noise trader, which is explored in detail in Section~\ref{section:welfare} and Appendix~\ref{section:exogenous}.

In the other extreme case when $\kappa\rightarrow 0$, the endogenous noise tracker places no emphasis on tracking the noisy target and instead only cares about maximizing expected wealth.  In this case, the endogenous noise tracker corresponds to a risk-neutral investor.  When a risk-neutral investor trades amongst the risk-averse utility-maximizing investors, the resulting equilibrium is trivial:  the risk-neutral investor optimally holds all $\Sigma$ shares, the risk-averse investors optimally hold no shares, and the equilibrium stock price is given by $S_t = D_t$ for all $t\in[0,1]$.
}

To prove Theorem~\ref{thm:radner-endogenous}, we first obtain the existence and uniqueness result for the system of ODEs in the above theorem.

\begin{lemma}\label{en_core_ode_lemma}
Let $a, \sigma_D^2, \kappa>0$. Then, the following two-dimensional initial value problem has a unique solution for $t\in [0,\infty)$:
\begin{equation}
\begin{split}\label{en_core_ode}
z_1'(t) &=\left(\tfrac{2a\kappa \sigma_D^2}{2I \kappa + a \sigma_D^2} \right) t -\tfrac{4 a \kappa  I \sigma_{Y'}^2 z_1(t)z_2(t) }{a \sigma_D^2+a \sigma_{Y'}^2  z_1(t)^2+2 \kappa  I},\quad z_1(0)=0,\\
z_2'(t)&= \tfrac{2 \kappa   \, z_1(t)}{a \sigma_D^2+2 \kappa  I}- \tfrac{2 a \sigma_{Y'}^2  \left(\left(a \sigma_D^2+2 \kappa  I\right)^2+a \sigma_{Y'}^2 z_1(t)^2 \left(a \sigma_D^2+4 \kappa  I \right)\right)z_2(t)^2}{\left(a \sigma_D^2+a \sigma_{Y'}^2 z_1(t)^2+2 \kappa  I \right)^2} ,\quad z_2(0)=0.
\end{split}
\end{equation}
{\color{black}The functions $z_1$ and $z_2$ are strictly positive on $(0,\infty)$.}
\end{lemma}

\begin{proof}
The local Lipschitz structure of \eqref{en_core_ode} produces a unique solution around the point $t=0$ by the Picard-Lindel\"of theorem (see, e.g., Theorem II.1.1 in Hartman~\cite{H02}). Furthermore, there exists a maximal interval of existence $0\in (\underline{t},\ot)\subseteq [-\infty,\infty]$ and we will argue by contradiction to show that $\ot=\infty$. 

Suppose that $\bar t<\infty$. By taking derivatives through \eqref{en_core_ode} and inserting the initial values $z_1(0)=z_2(0)=0$ into these derivatives at $t=0$, we find
\begin{align}
&z_1(0)=z_1'(0)=0<z_1''(0)=\tfrac{2a\kappa \sigma_D^2}{2I \kappa + a \sigma_D^2},\\
&z_2(0)=z_2'(0)=z_2''(0)=0<z_2'''(0)=\tfrac{4a\kappa^2 \sigma_D^2}{(2I \kappa + a \sigma_D^2)^2}.
\end{align}
The above equalities and inequalities imply that there exists $\epsilon>0$ such that $z_1(t)>0$ and $z_2(t)>0$ for $t\in (0,\epsilon)$. 

To reach a contradiction with $\ot<\infty$, we define $t_0$ as
 \begin{align}
 t_0:=\inf\{t>0: z_1(t)=0 \textrm{  or  } z_2(t)=0 \},
 \end{align}
which is the first time ($z_1$ or $z_2$) reaches zero strictly after time $t=0$. According to the previous argument, there exists $\epsilon>0$ such that $t_0\geq \epsilon$.

If $t_0<\bar t$, then there are two possibilities: (i) In case $z_1(t_0)=0$, then \eqref{en_core_ode} gives $z_1'(t_0)=\left(\frac{2a\kappa \sigma_D^2}{2I \kappa + a \sigma_D^2}\right)t_0>0$, which contradicts the definition of $t_0$. (ii) In case $z_2(t_0)=0$ (and $z_1(t_0)>0$ by (i) above), then \eqref{en_core_ode} gives $z_2'(t_0)=\frac{2\kappa \, z_1(t_0)}{2I \kappa + a \sigma_D^2}>0$, which contradicts the definition of $t_0$. Therefore, we conclude that $z_1(t)>0$ and $z_2(t)>0$ for $t\in (0,\bar t)$. This observation and \eqref{en_core_ode} imply
\begin{align}\label{eqn:z-bounds}
0<z_1(t)< \left(\tfrac{a\kappa \sigma_D^2}{2I \kappa + a \sigma_D^2} \right) t^2, \quad 0<z_2(t)<  \left(\tfrac{2a\kappa^2 \sigma_D^2}{3(2I \kappa + a \sigma_D^2)^2} \right) t^3, \quad \textrm{for} \quad t\in (0,\bar t).
\end{align}
The above boundedness properties contradicts $\bar t<\infty$. {\color{black}Finally, since $\bar t = \infty$, we have that $z_1$ and $z_2$ are strictly positive on $(0,\infty)$ by~\eqref{eqn:z-bounds}.}
\end{proof}

The utility-maximizing investors' value functions are conjectured to have the form \eqref{eqn:val-fn-form}-\eqref{eqn:term-cond}. We also conjecture the form of the noise tracker's value function $V_N$ with state processes time $t$, wealth $x$, $Y$, and $Y'$, and we search for smooth functions $f_1$, $f_2$, $f_3$, $f_{22}$, $f_{23}$, $f_{33}:[0,1]\rightarrow\R$ such that
\begin{equation}\label{eqn:tracker-value-fn-form}
  V_N(t,x,Y,Y') = x+f_1(t)+f_2(t)Y+f_3(t)Y'+ f_{22}(t)Y^2+ f_{23}(t)Y Y'+ f_{33}(t)Y'^2,
\end{equation}
with $f_1(1)=f_2(1)=f_3(1)=f_{22}(1)=f_{23}(1)=f_{33}(1)=0$.
The investors' HJB equations and the market clearing condition produce the ODE system for the coefficient functions in \eqref{eqn:val-fn-form} and \eqref{eqn:tracker-value-fn-form}, and $(\alpha,\beta,\mu)$ in \eqref{R11}:
\begin{align}\label{eqn:endogenous-ODEs}
\begin{split}
g'_1(t) &= \tfrac{a \sigma_{Y'}^2 g_{3}(t)^2 \left(a \sigma_{Y'}^2 \beta (t)^2 \left(a \sigma_D^2+4 I \kappa\right)+\left(a \sigma_D^2+2 I \kappa\right)^2\right)}{2 \left(a \sigma_D^2+a \sigma_{Y'}^2 \beta (t)^2+2 I \kappa\right)^2}\\
&\quad + \tfrac{2 a \Sigma \kappa \left( a \sigma_{Y'}^2 g_3(t) \beta(t) - \kappa \Sigma \right)  \left(\sigma_D^2+\sigma_{Y'}^2 \beta (t)^2\right)}{ \left(a \sigma_D^2+a \sigma_{Y'}^2 \beta (t)^2+2 I \kappa\right)^2} - \sigma_{Y'}^2 g_{33}(t),\\
g_2'(t)&= \tfrac{a\sigma_{Y'}^2 g_{23}(t) \left(g_{3}(t) \left(a \sigma_{Y'}^2 \beta (t)^2 \left(a \sigma_D^2+4 I \kappa\right)+\left(a \sigma_D^2+2 I \kappa\right)^2\right)+2 a \Sigma  \beta (t) \kappa \left(\sigma_D^2+\sigma_{Y'}^2 \beta (t)^2\right)\right)}{\left(a \sigma_D^2+a \sigma_{Y'}^2 \beta (t)^2+2 I \kappa\right)^2}\\
&\quad-\tfrac{2a \kappa \left(\sigma_D^2+\sigma_{Y'}^2 \beta (t)^2\right) \left(a \sigma_{Y'}^2 g_{3}(t) \beta (t)-2 \Sigma  \kappa\right)}{\left(a \sigma_D^2+a \sigma_{Y'}^2 \beta (t)^2+2 I \kappa\right)^2}, \\
g_3'(t) &= \tfrac{2 a \sigma_{Y'}^2 g_{33}(t) \left(g_{3}(t) \left(a \sigma_{Y'}^2 \beta (t)^2 \left(a \sigma_D^2+4 I \kappa\right)+\left(a \sigma_D^2+2 I \kappa\right)^2\right)+2 a \Sigma  \beta (t) \kappa \left(\sigma_D^2+\sigma_{Y'}^2 \beta (t)^2\right)\right)}{\left(a \sigma_D^2+a \sigma_{Y'}^2 \beta (t)^2+2 I \kappa\right)^2} -g_{2}(t),\\
g'_{22}(t) &= \tfrac{a \sigma_{Y'}^2 g_{23}(t)^2 \left(a \sigma_{Y'}^2 \beta (t)^2 \left(a \sigma_D^2+4 I \kappa\right)+\left(a \sigma_D^2+2 I \kappa\right)^2\right)}{2 \left(a \sigma_D^2+a \sigma_{Y'}^2 \beta (t)^2+2 I \kappa\right)^2}  - \tfrac{2 a \kappa \left( a \sigma_{Y'}^2 g_{23}(t) \beta (t) + \kappa \right) \left(\sigma_D^2+\sigma_{Y'}^2 \beta (t)^2\right)}{ \left(a \sigma_D^2+a \sigma_{Y'}^2 \beta (t)^2+2 I \kappa\right)^2},\\
g_{23}'(t)&= \tfrac{2a \sigma_{Y'}^2 g_{33}(t)\left(g_{23}(t) \left(a \sigma_{Y'}^2 \beta (t)^2 \left(a \sigma_D^2+4 I \kappa\right)+\left(a \sigma_D^2+2 I \kappa\right)^2\right)-2a \kappa \beta(t)\left( \sigma_D^2+ \sigma_{Y'}^2 \beta(t)^2 \right) \right)}{ \left(a \sigma_D^2+a \sigma_{Y'}^2 \beta (t)^2+2 I \kappa\right)^2} - 2 g_{22}(t),\\
g_{33}'(t) &= \tfrac{2 a \sigma_{Y'}^2 g_{33}(t)^2 \left(a \sigma_{Y'}^2 \beta (t)^2 \left(a \sigma_D^2+4 I \kappa\right)+\left(a \sigma_D^2+2 I \kappa\right)^2\right)}{\left(a \sigma_D^2+a \sigma_{Y'}^2 \beta (t)^2+2 I \kappa\right)^2}-g_{23}(t),\\
g_{1}(1)&=g_{2}(1)=g_{3}(1)=g_{22}(1)=g_{23}(1)=g_{33}(1)=0,\\
f_{1}'(t) &= -\tfrac{\kappa \left(a \sigma_{Y'}^2 \beta (t) (I g_{3}(t)+\Sigma  \beta (t))+a \Sigma  \sigma_D^2\right)^2}{\left(a \sigma_D^2+a \sigma_{Y'}^2 \beta (t)^2+2 I \kappa\right)^2}-\sigma_{Y'}^2 f_{33}(t),\\
f_{2}'(t) &= -\tfrac{2 a I \kappa \left(a \sigma_{Y'}^2 g_{23}(t) \beta (t)+2 \kappa\right) \left(\sigma_{Y'}^2 \beta (t) (I g_{3}(t)+\Sigma  \beta (t))+\Sigma  \sigma_D^2\right)}{\left(a \sigma_D^2+a \sigma_{Y'}^2 \beta (t)^2+2 I \kappa\right)^2},\\
f_{3}'(t) &= -\tfrac{4 a^2 I \sigma_{Y'}^2 g_{33}(t) \beta (t) \kappa \left(\sigma_{Y'}^2 \beta (t) (I g_{3}(t)+\Sigma  \beta (t))+\Sigma  \sigma_D^2\right)}{\left(a \sigma_D^2+a \sigma_{Y'}^2 \beta (t)^2+2 I \kappa\right)^2}-f_{2}(t),\\
f_{22}'(t)&= \tfrac{a\kappa \left( \sigma_{Y'}^2 \beta(t)(\beta(t)-I g_{23}(t)) + \sigma_D^2  \right) \left( a \sigma_{Y'}^2 \beta(t)(M g_{23}(t)+\beta(t))+ a\sigma_D^2 +4I \kappa  \right)   }{\left(a \sigma_D^2+a \sigma_{Y'}^2 \beta (t)^2+2 I \kappa\right)^2}, \\
f_{23}'(t) &= -\tfrac{4 a I^2 \sigma_{Y'}^2 g_{33}(t) \beta (t) \kappa \left(a \sigma_{Y'}^2 g_{23}(t) \beta (t)+2 \kappa\right)}{\left(a \sigma_D^2+a \sigma_{Y'}^2 \beta (t)^2+2 I \kappa\right)^2} - 2f_{22}(t),\\
f_{33}'(t) &= -\tfrac{4 a^2 I^2 \sigma_{Y'}^4  g_{33}(t)^2 \beta(t)^2 \kappa }{\left(a \sigma_D^2+a \sigma_{Y'}^2 \beta (t)^2+2 I \kappa\right)^2} - f_{23}(t),\\
f_{1}(1)&=f_{2}(1)=f_{3}(1)=f_{22}(1)=f_{23}(1)=f_{33}(1)=0,
  \end{split}
\end{align}
where
\begin{align}\label{eqn:endogenous-abm}
\begin{split}
\alpha'(t)&=-\tfrac{2 a \kappa \left(\sigma_{Y'}^2 \beta (t) (\beta (t)-I g_{23}(t))+\sigma_D^2\right)}{a \sigma_D^2+a \sigma_{Y'}^2 \beta (t)^2+2 I \kappa},\quad \alpha(1)=0,\\
\beta'(t)&=\tfrac{4 a I \sigma_{Y'}^2 g_{33}(t) \beta (t) \kappa}{a \sigma_D^2+a \sigma_{Y'}^2 \beta (t)^2+2 I \kappa} - \alpha(t),\quad \beta(1)=0,\\
\mu'(t)&=\tfrac{2 a \kappa \left(\sigma_{Y'}^2 \beta (t) (I g_{3}(t)+\Sigma  \beta (t))+\Sigma  \sigma_D^2\right)}{a \sigma_D^2+a \sigma_{Y'}^2 \beta (t)^2+2 I \kappa},\quad \mu(1)=0.
\end{split}
\end{align}
{\color{black}In Appendix A, we explain more about the derivation of the ODE system.}

This ODE system will be used to verify optimality of the equilibrium trading strategies in the proof of Theorem~\ref{thm:radner-endogenous}.  Lemma~\ref{ode_existence_en} below establishes the existence and uniqueness of a smooth solution to the system~\eqref{eqn:endogenous-ODEs}-\eqref{eqn:endogenous-abm}.
\begin{lemma}\label{ode_existence_en}
Let $a, \sigma_D^2, \kappa>0$. Then, there exists a unique smooth solution to the coupled system of ODEs \eqref{eqn:endogenous-ODEs}-\eqref{eqn:endogenous-abm} for $t\in [0,1]$.  {\color{black}The functions $\beta$ and $g_{33}$ are strictly positive on $[0,1)$.}
\end{lemma}
\begin{proof}
Once we show the existence of a solution of \eqref{eqn:endogenous-ODEs}-\eqref{eqn:endogenous-abm}, the local Lipschitz structure of the system ensures the uniqueness of the solution. To show the existence, we construct  a solution using Lemma \ref{en_core_ode_lemma}. In \eqref{eqn:endogenous-ODEs}-\eqref{eqn:endogenous-abm}, we observe that the ODEs for $\alpha, \beta, \mu, g_{33}, g_{23}, g_{22}, g_{3}, g_{2}$ do not depend on $g_1, f_1, f_2, f_3, f_{22}, f_{23}, f_{33}$. Hence, we first define $\alpha, \beta, \mu, g_{33}, g_{23}, g_{22}, g_{3}, g_{2}$ in terms of $z_1$ and $z_2$ in Lemma \ref{en_core_ode_lemma}:
\begin{equation}
\begin{split}\label{en_ode_solutions}
\alpha(t)&:=\tfrac{2a \kappa \sigma_D^2}{2I \kappa + a \sigma_D^2}(1-t),\\
\beta(t)&:=z_1(1-t),\\
\mu(t)&:=-\tfrac{2a \kappa \sigma_D^2 \Sigma }{2I \kappa + a \sigma_D^2}(1-t),\\
g_{33}(t)&:=z_2(1-t),\\
g_{23}(t)&:=\tfrac{2\kappa }{2I \kappa + a \sigma_D^2}\, z_1(1-t),\\
g_{22}(t)&:=\tfrac{2a \kappa^2 \sigma_D^2}{(2I \kappa + a \sigma_D^2)^2}(1-t),\\
g_{3}(t)&:=-\tfrac{2\kappa \Sigma }{2I \kappa + a \sigma_D^2}\,  z_1(1-t),\\
g_{2}(t)&:=-\tfrac{4a \kappa^2 \sigma_D^2 \Sigma }{(2I \kappa + a \sigma_D^2)^2}(1-t).
\end{split}
\end{equation}
By explicit computations using \eqref{en_core_ode}, we can check that \eqref{en_ode_solutions} satisfies the ODEs for $\alpha, \beta, \mu, g_{33}, g_{23}, g_{22}, g_{3}, g_{2}$ in \eqref{eqn:endogenous-ODEs}-\eqref{eqn:endogenous-abm}.\footnote{
{\color{black}
A more detailed explanation is provided here. The definitions of $\beta$, $g_{23}$,  and $g_3$ in \eqref{en_ode_solutions} imply that
\begin{align}
g_{23}(t)&=\tfrac{2\kappa }{2I \kappa + a \sigma_D^2}\, \beta(t),\quad g_{3}(t)=-\tfrac{2\kappa \Sigma }{2I \kappa + a \sigma_D^2}\,  \beta(t).\label{g_beta}
\end{align}
The definition of $\alpha$ in \eqref{en_ode_solutions} and $g_{23}$ in \eqref{g_beta} validate $\alpha$'s ODE in \eqref{eqn:endogenous-abm}.
The definitions of $\beta$, $\alpha$, and $g_{33}$ in \eqref{en_ode_solutions} and the expression for $z_1'$ in \eqref{en_core_ode} validate $\beta$'s ODE in \eqref{eqn:endogenous-abm}.
The definition of $\mu$ in \eqref{en_ode_solutions} and $g_3$ in \eqref{g_beta} validate $\mu$'s ODE in \eqref{eqn:endogenous-abm}.
The definitions of $\beta$, $g_{33}$, and $g_{23}$ in \eqref{en_ode_solutions} and  $z_2'$ in \eqref{en_core_ode} validate $g_{33}$'s ODE in \eqref{eqn:endogenous-ODEs}.
The definitions of $\beta$, $\alpha$, $g_{33}$, $g_{23}$, and $g_{22}$ in \eqref{en_ode_solutions} and $z_1'$ in \eqref{en_core_ode} validate $g_{23}$'s ODE in \eqref{eqn:endogenous-ODEs}.
The definition of $g_{22}$ in \eqref{en_ode_solutions} and relation \eqref{g_beta} validate the expression of $g_{22}'$ in \eqref{eqn:endogenous-ODEs}. 
The definitions for $\alpha$, $\beta$, $g_{22}$, $g_3$, and $g_2$ in \eqref{en_ode_solutions} and $z_1'$ in \eqref{en_core_ode} validate $g_3$'s ODE in \eqref{eqn:endogenous-ODEs}. Finally, the definition of $g_2$ in \eqref{en_ode_solutions} and the expressions in \eqref{g_beta} validate the $g_2$'s ODE in \eqref{eqn:endogenous-ODEs}.
}
}
 {\color{black}Moreover, Lemma~\ref{en_core_ode_lemma} ensures that $z_1$ and $z_2$ are strictly positive on $(0,\infty)$, which proves that $\beta$ and $g_{33}$ are strictly positive on $[0,1)$.}

Given \eqref{en_ode_solutions}, the ODEs for $g_1, f_1, f_2, f_3, f_{22}, f_{23}, f_{33}$ in \eqref{eqn:endogenous-ODEs}-\eqref{eqn:endogenous-abm} become a linear system of ODEs and we have the following explicit solutions:
\begin{equation}
\begin{split}\label{en_ode_solutions_f}
g_1(t)&:=\tfrac{2a \kappa^2 \sigma_D^2 \Sigma^2 }{(2I \kappa + a \sigma_D^2)^2}(1-t) + \sigma_{Y'}^2 \int_{t}^1 g_{33}(s)\, ds,\\
f_2(t)&:= \tfrac{4a I \kappa^2 \sigma_D^2 \Sigma }{(2I \kappa + a \sigma_D^2)^2}(1-t) ,\\
f_3(t)&:= \tfrac{2a I \kappa^2 \sigma_D^2 \Sigma }{(2I \kappa + a \sigma_D^2)^2}(1-t)^2 + \int_{t}^1 \tfrac{4a^2 I \kappa \Sigma \sigma_D^2 \sigma_{Y'}^2 g_{33}(s)\beta(s)}{(2I \kappa+a \sigma_D^2)\left(2I\kappa + a\sigma_D^2+ a\sigma_{Y'}^2 \beta(s)^2\right)}  \, ds   ,\\
f_{22}(t)&:=\left( \tfrac{4 I^2 \kappa^3}{(2I \kappa + a \sigma_D^2)^2} -\kappa \right) (1-t),\\
f_{23}(t)&:=\left( \tfrac{4 I^2 \kappa^3}{(2I \kappa + a \sigma_D^2)^2} -\kappa \right) (1-t)^2   + \int_{t}^1 \tfrac{8a I^2 \kappa^2  \sigma_{Y'}^2 g_{33}(s)\beta(s)}{(2I \kappa+a \sigma_D^2)\left(2I\kappa + a\sigma_D^2+ a\sigma_{Y'}^2 \beta(s)^2\right)}  \, ds ,\\
f_{33}(t)&:=\int_t^1 f_{23}(s)\, ds +    \int_{t}^1 \tfrac{4a^2 I^2 \kappa  \sigma_{Y'}^4 g_{33}(s)^2\beta(s)^2}{\left(2I\kappa + a\sigma_D^2+ a\sigma_{Y'}^2 \beta(s)^2\right)^2}  \, ds, \\
f_1(t)&:=  \tfrac{a^2 \kappa \sigma_D^4 \Sigma^2 }{(2I \kappa + a \sigma_D^2)^2}(1-t) + \sigma_{Y'}^2 \int_{t}^1 f_{33}(s)\, ds.\\
\end{split}
\end{equation}
Note that the integrals above are all finite due to $a, \sigma_D^2, \kappa>0$.
\end{proof}

\begin{proof}[Proof of Theorem \ref{thm:radner-endogenous}]\ 

Let $(\alpha,\beta,\mu,g_1,g_2,g_3,g_{22},g_{23},g_{33}, f_1, f_2, f_3, f_{22}, f_{23}, f_{33})$ be the unique smooth solution to the coupled system of ODEs \eqref{eqn:endogenous-ODEs}-\eqref{eqn:endogenous-abm}, whose expression is given in \eqref{en_ode_solutions}-\eqref{en_ode_solutions_f}.  {\color{black}By Lemma~\ref{ode_existence_en}, the functions $\beta$ and $g_{33}$ are strictly positive on $[0,1)$.}

  First, we prove verification for the utility-maximizing exponential investors via a duality approach.  
     For the function $V$ given in \eqref{eqn:val-fn-form}, we let the process $\widehat V$ be given by $\widehat V_t:=V(t,X^{\hat\theta}_t,Y_t,Y'_t)$, $t\in[0,1]$, where $\hat\theta$ is defined in~\eqref{def:hat-theta}.  We drop the subscript $j$ from $\hat\theta$ since all utility-maximizing investors are identical. For notational simplicity, we denote $\widehat X_t:= X^{\hat\theta}_t$.  We define the measure $\widehat\bQ$ by
  $$
    \frac{d\widehat\bQ}{d\P}:= \frac{\widehat V_1}{\widehat V_0},
  $$ 
  and we reason that $\widehat\bQ$ is a probability measure.  By \eqref{eqn:val-fn-form}, the terminal conditions \eqref{eqn:term-cond}, and the ODE system \eqref{eqn:endogenous-ODEs}-\eqref{eqn:endogenous-abm}, the dynamics of $\widehat V$ are given by
  $$
    d\widehat V_t = -a \widehat V_t\left(\hat\theta_t dD_t + \left(\beta(t) \hat\theta_t + g_3(t)+2Y'_tg_{33}(t)+Y_t g_{23}(t)\right)dY'_t\right).
  $$
    Since $\hat\theta_t$ is affine in $(Y_t,Y_t')$, the functions $\beta$, $g_3$, $g_{23}$, and $g_{33}$ are continuous functions of $t$, and $Y_t$ is a progressively measurable functional of $Y'$, we apply Corollary 3.5.16 of Karatzas and Shreve~\cite{KS91} to show that $(\widehat V_t)_{t\in [0,1]}$ is a martingale under $\P$. Thus, $\widehat\bQ$ is a probability measure.  Since $\widehat V_1 = -e^{-a \widehat X_1}$, we also have that $\widehat V_0 = \E\left[-e^{-a\widehat X_1}\right]$.
  
  Next, we show that $\widehat X$ is a $\widehat\bQ$-martingale by checking that it is a $\widehat\bQ$-local martingale and $\E^{\widehat\bQ}\left[\int_0^1 \hat\theta_t^2 d\left<S\right>_t\right]<\infty$.  Under $\widehat\bQ$, $\widehat D$ and $\widehat Y'$ are  Brownian motions, where
  \begin{align}
  \begin{split}\label{def:Qhat-BM}
    d\widehat{D}_t &:= dD_t + a\sigma_D^2 \hat\theta_t dt, \quad \widehat D_0:=0,\\
    d\widehat{Y}'_t &:= dY'_t + a\sigma_{Y'}^2 \left(\beta(t)\hat\theta_t + g_3(t)+2Y'_t g_{33}(t)+Y_t g_{23}(t)\right)dt, \quad \widehat Y'_0:=0.
   \end{split}
  \end{align}

  By \eqref{def:hat-theta} and \eqref{R11}, the dynamics of $\widehat X$ are given by
  $$
    d\widehat X_t = \hat\theta_t dS_t 
    = 
    \hat\theta_t\left(d\widehat D_t + \beta(t)d\widehat{Y}'_t\right).
  $$
 
Since $\beta(t)$ is bounded in $t\in [0,1]$, it suffices to show that $\E^{\widehat\bQ}\left[\int_0^1 \hat\theta^2_t dt\right]<\infty$ in order to prove that the $\widehat\bQ$-local martingale, $\widehat X$, is a $\widehat\bQ$-martingale.  Moreover, since $\hat\theta_t$ is affine in $(Y_t,Y'_t)$ and the coefficient functions that appear in the ODE system are all bounded in $t\in [0,1]$, showing that $\E^{\widehat\bQ}\left[\int_0^1 (Y'_t)^2 dt\right]<\infty$ is sufficient to prove the $\widehat{\bQ}$-martingale property of $\widehat X$. 
    To this end, we define stopping times
  $$
    \tau_k:=\inf\left\{t\geq 0:\ |Y'_t|\geq k\right\}\wedge 1, \quad k\geq 1.
  $$
  We appeal to the definition of $\widehat Y'$ in \eqref{def:Qhat-BM} to see that there exist constants $C_1, C_2\geq 0$ that are independent of $t$ and $k$ such that for all $k\geq 1$,  
  $$
    \E^{\widehat\bQ}\left[\left(Y'_{t\wedge\tau_k}\right)^2\right]
    \leq C_1 + C_2\,\E^{\widehat\bQ}\left[\int_0^t\left(Y'_{u\wedge\tau_k}\right)^2du\right], \quad t\in[0,1].
  $$
  Gronwall's inequality implies that for all $k\geq 1$,
  $$
    \E^{\widehat\bQ}\left[\left(Y'_{t\wedge\tau_k}\right)^2\right] \leq C_1 e^{C_2 t},
    \quad t\in[0,1].
  $$
  By Fatou's Lemma, we have
  $$
    \E^{\widehat\bQ}\left[\left(Y'_{t}\right)^2\right]
    \leq \liminf_{k\rightarrow\infty} \E^{\widehat\bQ}\left[\left(Y'_{t\wedge\tau_k}\right)^2\right] \leq C_1 e^{C_2 t}, \quad t\in[0,1],
  $$
  and thus,  
  $$
    \E^{\widehat\bQ}\left[\int_0^1 (Y'_t)^2\, dt\right] 
    \leq \int_0^1 C_1 e^{C_2t}dt < \infty.
  $$
  
  The exponential investors have utility functions $U(x):=-\exp(-ax)$, $x\in\R$.  The Fenchel-Legendre transform of the function $-U(-x)$ is given by
  $$
    \widetilde U(y):= \sup_{x\in\R}\left\{U(x)-xy\right\} = -\frac{y}{a}\left(1-\log\frac{y}{a}\right),\quad y>0.
  $$
  Therefore, for any admissible $\theta\in\sA(\widehat\bQ)$ with an associated wealth process $X^\theta$ beginning with initial wealth $X^\theta_0 = \widehat X_0$, we have
  \begin{align*}
    \E&\left[U\left(X^\theta_1\right)\right]\\
    &\leq \E\left[\widetilde U\left(a e^{-a \widehat X_1}\right) + a e^{-a \widehat X_1}X^\theta_1\right]\\
    &= \E\left[e^{-a\widehat X_1}\left(-a\widehat{X}_1-1\right)\right]+ a \E^{\widehat\bQ}[X^\theta_1]\cdot \E\left[e^{-a\widehat{X}_1}\right]
      \quad \text{by definition of $\widetilde U$ and $\widehat\bQ$}\\
    &= \E\left[U\left(\widehat X_1\right)\right]-a\E^{\widehat\bQ}[\widehat X_1]\cdot\E\left[e^{-a\widehat{X}_1}\right]+ a \E^{\widehat\bQ}[X^\theta_1]\cdot \E\left[e^{-a\widehat{X}_1}\right]\\
    &\leq \E\left[U\left(\widehat X_1\right)\right] \quad
    \text{since $\widehat X$ is a $\widehat\bQ$-martingale and $\theta\in\sA(\widehat\bQ)$},
  \end{align*}
which shows that $\hat\theta$ is the optimal strategy for the utility-maximizing investors.

Second, we show that $\hat\theta_N$ in \eqref{def:tracker-hat-theta} is optimal for the noise tracker.  In addition to being adapted and measurable, $\hat\theta_N$ is square-integrable with $\E\left[\int_0^1\hat\theta_{N,t}^2\,dt \right] <\infty$, since $\beta$ and $g_{33}$ are continuous functions of time and $Y$ and $Y'$ are Gaussian.

Pointwise maximization under the expectation and integral in~\eqref{optproblem_tracker} lead to
$$
  \hat\theta_{N,t} = Y_t + \frac{1}{2\kappa}\Big(\mu'(t) + \alpha'(t)Y_t + (\beta'(t)+\alpha(t))Y'_t\Big).
$$
Plugging in $\mu$, $\alpha$, $\beta$, and $g_{33}$ as given in \eqref{en_ode_solutions} agrees with the desired formula for $\hat\theta_N$ in~\eqref{def:tracker-hat-theta}.

Finally, the form of the optimal strategies in \eqref{def:hat-theta} and \eqref{def:tracker-hat-theta} ensure the market clearing condition in \eqref{eqn:clearing2}.
\end{proof}

\section{Welfare implications}\label{section:welfare}
Welfare describes well being.  In equilibrium, aggregate welfare analysis is important because it helps to explain how an economy's set-up and its resulting equilibrium resource allocation impact the well being of its participants.

We ask, 
\begin{center}
  {\it How does an endogenous noise tracker impact aggregate welfare?}
\end{center}
To set up a model that can answer this question, we must consider two important points.  First, absolute levels of welfare are not informative on their own.  Only relative welfare is of interest when compared to a baseline model's welfare.  Second, noise traders with exogenous noise often exist in equilibrium to fulfill an unmet modeling need.  For example, without noise traders in Grossman and Stiglitz~\cite{GS80AER}, Kyle~\cite{K85E}, Hellwig~\cite{H80JET}, or Vayanos~\cite{V01JF}, the resulting equilibria would be uninteresting.

Taking {\color{black}the first of} these two points into consideration, we need to compare welfare in our equilibrium to welfare in a related baseline model.  {\color{black}The second point means that we need to choose a baseline model that is similar to the endogenous noise tracker model and rich enough to be nontrivial.} We chose a Radner equilibrium with an exogenous noise trader, similar to the above cited works, as our baseline.  Though both exogenous noise traders and Radner equilibria are standard modeling features, our particular baseline model with both features combined has not been studied before.  We introduce this model in the Appendix, where we also include an equilibrium existence proof.

{\color{black}The exogenous noise trader baseline model corresponds to the extreme case of the endogenous noise tracker model when $\kappa\rightarrow\infty$.  In this case, the endogenous noise tracker cares only about tracking the noisy target and not about expected wealth.  For the endogenous noise tracker, infinite $\kappa$ can be interpreted as optimally holding $\hat\theta_N = Y$ shares, while the exogenous noise trader is defined by holding $Y$ shares.}

Which type of noise trader do utility maximizers prefer --  exogenous noise trading or an endogenous noise tracker? Based on a welfare analysis of the two models, the answer is that it depends.  We measure the utility maximizers' aggregate welfare by the sum of their certainty equivalents,
$$
  \sum_{j=1}^I \text{CE}_j,
$$
where the certainty equivalents $\text{CE}_j$ are defined implicitly by
\begin{equation*}
  -e^{-a \text{CE}_j} = \sup_{\theta\in\sA(\widehat\bQ)}\E\left[-\exp\left(-a X^\theta_{1}\right)\right],\quad j=1,\ldots,I.
\end{equation*}
The aggregate welfare is defined by
\begin{align*}
  \sum_{j=1}^I \text{CE}_j
  &= S_0\left(\Sigma-Y_0\right) + I\Big(g_1(0)+g_2(0)Y_0\\
  &\quad+g_3(0)Y_0'+ g_{22}(0)Y_0^2+ g_{23}(0)Y_0 Y_0'+ g_{33}(0)Y_0'^2\Big),
\end{align*}
where the functions $g_1$, $g_2$, $g_3$, $g_{22}$, $g_{23}$, and $g_{33}$ come from the functional form of the exponential investors' value functions in~\eqref{eqn:val-fn-form}-\eqref{eqn:term-cond}.  For both the exogenous and endogenous noise models, these functions are formally described and defined in the proofs of Theorem~\ref{thm:radner-endogenous} and Theorem~\ref{thm:radner-exogenous}. 
In the $Y_0=Y'_0=0$ case, Proposition~\ref{prop:welfare} below quantifies the welfare difference between exogenous noise trading and endogenous noise tracking.  To distinguish between quantities in the two models, we let the superscripts $ex$ and $en$ represent exogenous and endogenous quantities, respectively.
\begin{proposition}\label{prop:welfare}  For $Y_0=Y'_0=0$, 
the endogenous noise aggregate welfare is greater than the exogenous noise aggregate welfare if and only if
\begin{align}\label{ineq:welfare}
  \Sigma^2 > -\frac{2I^2\sigma_{Y'}^2\left(a\sigma_D^2+2\kappa I\right)^2}{a^3\sigma_D^6} \int_0^1\left(g^{en}_{33}(u)-g^{ex}_{33}(u)\right)du,
\end{align}
where $g^{ex}_{33}$ and $g^{en}_{33}$ are the functions $g_{33}$ with exogenous and endogenous noise from Theorems~\ref{thm:radner-exogenous} and \ref{thm:radner-endogenous}, respectively.
The inequality \eqref{ineq:welfare} holds for a sufficiently large stock supply $\Sigma$.
\end{proposition}

\begin{proof}
We let the superscripts $ex$ and $en$ represent exogenous and endogenous quantities, respectively.  Theorems~\ref{thm:radner-exogenous} and \ref{thm:radner-endogenous} provide us with the existence of equilibria described by the systems of equations given in \eqref{eqn:exogenous-ODEs}-\eqref{eqn:exogenous-abm} and \eqref{eqn:endogenous-ODEs}-\eqref{eqn:endogenous-abm}, respectively.  Using the equilibrium clearing conditions from Definitions~\ref{def:exo-eq} and \ref{def:end-eq}, we see that
\begin{align*}
  \sum_{j=1}^I \text{CE}^{en}_j-\sum_{j=1}^I \text{CE}^{ex}_j
  &= \Sigma\left(S^{en}_0-S^{ex}_0\right) + I\left(g^{en}_1(0)-g^{ex}_1(0)\right)\\
  &= \Sigma\left(\mu^{en}(0)-\mu^{ex}(0)\right) +I\left(g^{en}_1(0)-g^{ex}_1(0)\right).
\end{align*}
The system of equations \eqref{eqn:exogenous-ODEs}-\eqref{eqn:exogenous-abm} decouples so that
\begin{align*}
  (g^{ex}_1)'(t)&= -\frac{a \Sigma ^2 \sigma_D^2}{2 I^2}-\sigma_{Y'}^2 g^{ex}_{33}(t),\\
  \mu^{ex}(t)&=\frac{a \Sigma  \sigma_D^2}{I}(t-1).
\end{align*}
and the system of equations \eqref{eqn:endogenous-ODEs}-\eqref{eqn:endogenous-abm} decouples so that
\begin{align*}
  (g^{en}_1)'(t)&= -\frac{2 a \kappa^2 \Sigma ^2 \sigma_D^2}{\left(a \sigma_D^2+2 \kappa I\right)^2}-\sigma_{Y'}^2 g^{en}_{33}(t),\\
  \mu^{en}(t)&=\frac{2 a\kappa \Sigma   \sigma_D^2}{a \sigma_D^2+2 \kappa I}(t-1).
\end{align*}
These calculations show us that in the exogenous case,
$$
  \Sigma(\mu^{ex})'(0) + I (g^{ex}_1)'(0) = \frac{a \Sigma ^2 \sigma_D^2}{2 I}-g^{ex}_{33}(0) I \sigma_{Y'}^2,
$$
and in the endogenous case, we have
$$
  \Sigma(\mu^{en})'(0) + I (g^{en}_1)'(0) = \frac{2 a \kappa \Sigma ^2 \sigma_D^2 \left(a \sigma_D^2+\kappa I\right)}{\left(a \sigma_D^2+2 \kappa I\right)^2}-g^{en}_{33}(0) I \sigma_{Y'}^2.
$$
Therefore, the welfare difference is calculated as
\begin{align}\label{eqn:welfare-difference}
  \sum_{j=1}^I \text{CE}^{en}_j-\sum_{j=1}^I \text{CE}^{ex}_j
  &= \frac{a^3 \Sigma ^2 \sigma_D^6}{2 I \left(a \sigma_D^2+2 \kappa I\right)^2}
    + I\sigma_{Y'}^2\int_0^1 \left(g^{en}_{33}(u)-g^{ex}_{33}(u)\right)du.
\end{align}
Using this, we conclude that $  \sum_{j=1}^I \text{CE}^{en}_j> \sum_{j=1}^I \text{CE}^{ex}_j$ if and only if \eqref{ineq:welfare} holds.

Since the two-ODE systems -- $(g^{ex}_{33},\beta^{ex})$ in Theorem~\ref{thm:radner-exogenous} and $(g^{en}_{33},\beta^{en})$ in Theorem~\ref{thm:radner-endogenous} -- do not depend on $\Sigma$, the right-hand side in \eqref{ineq:welfare} does not depend on $\Sigma$. Therefore, we conclude that \eqref{ineq:welfare} holds for large enough $\Sigma$.
\end{proof}


In words, Proposition~\ref{prop:welfare} implies that for a sufficiently large stock supply $\Sigma$, the endogenous noise tracker model is preferable for the exponential investors compared to the exogenous noise trader model when $Y_0=Y'_0=0$.


\begin{figure}[t]
	\begin{center}
	$
			\begin{array}{cc}
			\includegraphics[width=0.4\textwidth]{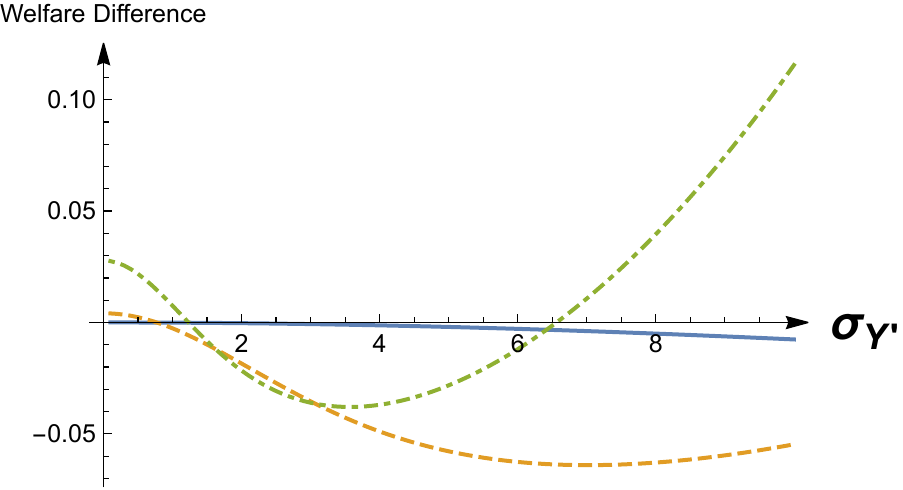} &
			\includegraphics[width=0.4\textwidth]{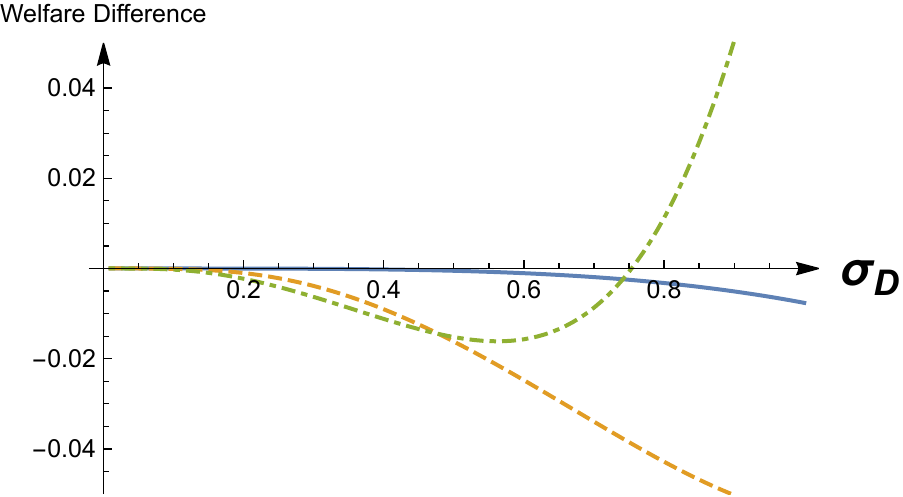}  \\
			\bigskip
			I=10, \, \sigma_D=1,\, \kappa=5 & I=10, \sigma_{Y'}=10, \, \kappa=5 \\ 
			\includegraphics[width=0.4\textwidth]{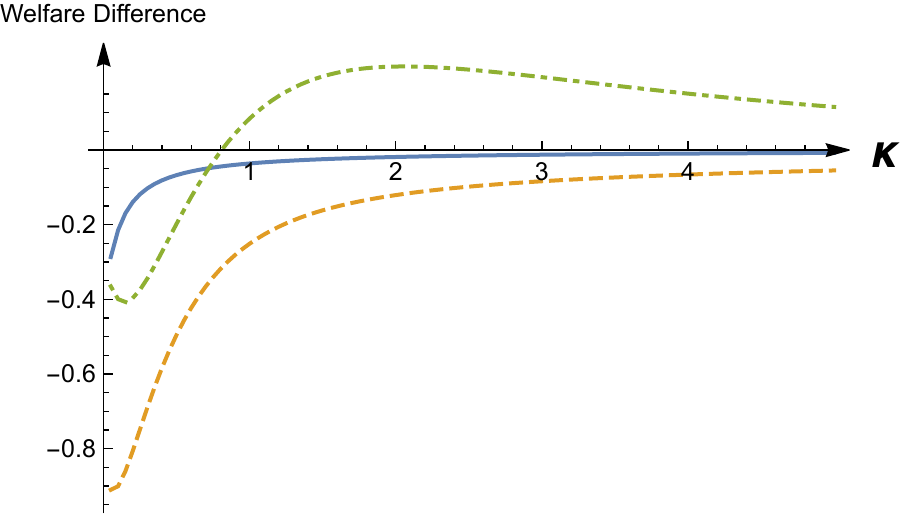} &
			\includegraphics[width=0.4\textwidth]{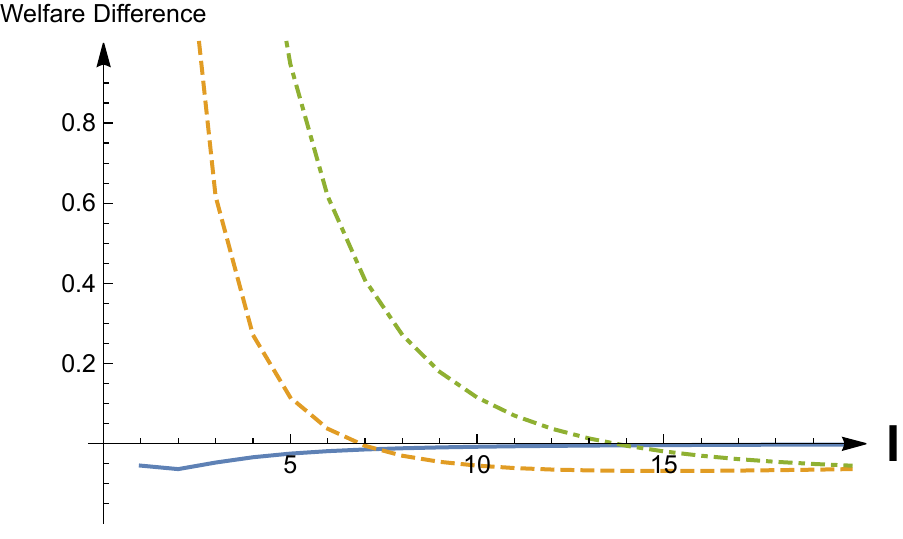} \\ 
			I=10, \, \sigma_D=1,\, \sigma_{Y'}=10 & \sigma_D=1, \sigma_{Y'}=10, \, \kappa=5
			\end{array}$
	\end{center}
	\caption{Graphs of welfare difference for varying $\sigma_{Y'}, \sigma_D, \kappa$, and $I$. The common set of parameters are $Y_0=Y'_0=0$, $\Sigma=1$ and $a=1$ (\textcolor{myblue}{\bf ---}), $10$ (\textcolor{myorange}{-\,-\,-}), $20$ (\textcolor{mygreen}{- $\cdot$ -}).  }
	\label{figure:welfare}
\end{figure}

Figure~\ref{figure:welfare} plots the welfare difference across a common set of parameters:  $Y_0=Y'_0=0$, $\Sigma=1$, $I=10$, $\sigma_D=1$, $\sigma_{Y'}=10$, $\kappa=5$, and $a\in\{1,10,20\}$.  By varying these parameters, we observe non-monotone behavior.  Amongst the two-dimensional parameter choices $(\sigma_{Y'},a)$, $(\sigma_{D},a)$, $(\kappa,a)$, and $(I,a)$, we see similar non-trivial, non-monotone dependencies.  The behavior exhibited in these plots shows that the modeling choices involved with noise traders have complex effects on welfare.  Our model suggests that care should be taken when drawing conclusions about welfare effects when noise traders are present.

\appendix

{\color{black}

\section{Derivation of the ODE system}\label{app:A}

In this appendix, we explain how to derive the ODE system in \eqref{eqn:endogenous-ODEs}-\eqref{eqn:endogenous-abm}.
Using the expression of $S_t$ in \eqref{R11}, we observe that $dX_t^\theta$ in \eqref{def:wealth} becomes
\begin{align}
dX_t^\theta = \theta_t \left( \mu'(t)+ (\alpha(t)+\beta'(t))Y'_t + \alpha'(t) Y_t  \right)dt + \theta_t \beta(t) dY'_t + \theta_t dD_t.\label{wealth2}
\end{align}
We apply Ito's formula to $V(t,X^\theta_t,Y_t,Y'_t)$ for the function $V$ in \eqref{eqn:val-fn-form} and obtain
\begin{align}
&\frac{dV(t,X^\theta_t,Y_t,Y'_t)}{-a V(t,X^\theta_t,Y_t,Y'_t)}= \Big(  
-\tfrac{a\left(\sigma_D^2+\sigma_{Y'}^2\beta(t)^2\right)}{2} \theta_t^2  \nonumber\\
&+ \left( \mu'(t) - a \sigma_{Y'}^2\beta(t)(g_{23}(t)Y_t + g_3(t))+\alpha'(t)Y_t + (\alpha(t)-2a\sigma_{Y'}^2 g_{33}(t)\beta(t)+ \beta'(t))Y'_t
 \right)\theta_t\nonumber\\
 & + g_1'(t)+g_2'(t)Y_t + g_3'(t) Y'_t+ g_{22}'(t)Y_t^2 + g_{23}'(t)Y_t Y'_t + g_{33}'(t) Y'^2_t \nonumber\\
& +g_2(t) Y'_t + 2 g_{22}(t) Y_t Y'_t + g_{23}(t)Y'^2_t + \tfrac{\sigma_{Y'}^2(2g_{33}(t)-a(g_3(t)+g_{23}(t) Y_t + 2g_{33}(t)Y'_t)^2)}{2}   \Big)dt \nonumber\\
& + \theta_t dD_t + \left( \theta_t \beta(t) + g_3(t) + g_{23}(t)Y_t + 2 g_{33}(t)Y'_t \right) dY'_t, \label{dV_expand}
\end{align}
where we use \eqref{wealth2} and \eqref{pp1}. Observe that the $dt$ term above is maximized when $\theta_t$ equals $ \theta_{j,t}$ below: 
\begin{align}
 \theta_{j,t}=\tfrac{\mu'(t) - a \sigma_{Y'}^2\beta(t)(g_{23}(t)Y_t + g_3(t))+\alpha'(t)Y_t + (\alpha(t)-2a\sigma_{Y'}^2 g_{33}(t)\beta(t)+ \beta'(t))Y'_t}{a(\sigma_D^2 + \sigma_{Y'}^2 \beta(t)^2)}. \label{opt_theta_1}
\end{align}
We also apply Ito's formula to $V_N(t,X^\theta_t,Y_t,Y'_t)$ for the function $V_N$ in \eqref{eqn:tracker-value-fn-form} and obtain
\begin{align}
&dV_N(t,X^\theta_t,Y_t,Y'_t)- \kappa (\theta_t-Y_t)^2 dt = \Big(  - \kappa (\theta_t-Y_t)^2 \nonumber\\
&+ Y'_t(f_2(t)+2f_{22}(t)Y_t+f_{23}(t)Y'_t)+ \sigma_{Y'}^2 f_{33}(t) + \theta_t (\alpha'(t)Y_t + (\alpha(t)+ \beta'(t))Y'_t+\mu'(t)) \nonumber\\
 & + f_1'(t)+f_2'(t)Y_t + f_3'(t) Y'_t+ f_{22}'(t)Y_t^2 + f_{23}'(t)Y_t Y'_t + f_{33}'(t) Y'^2_t   \Big)dt \nonumber\\
& + \theta_t dD_t + \left( \theta_t \beta(t) + f_3(t) + f_{23}(t)Y_t + 2 f_{33}(t)Y'_t \right) dY'_t, \label{dVN_expand}
\end{align}
where we use \eqref{wealth2} and \eqref{pp1}. Observe that the $dt$ term above is maximized when $\theta_t$ equals $ \theta_{N,t}$ below: 
\begin{align}
 \theta_{N,t}=\tfrac{2\kappa Y_t + \alpha'(t)Y_t + (\alpha(t)+ \beta'(t))Y'_t+\mu'(t)}{2\kappa}. \label{opt_theta_2}
\end{align}

We apply \eqref{opt_theta_1} and \eqref{opt_theta_2} to the market clearing condition in \eqref{eqn:clearing2}, and obtain the ODEs in \eqref{eqn:endogenous-abm}.

We substitute $ \theta_{j,t}$ in \eqref{opt_theta_1} for $\theta_t$ in \eqref{dV_expand}, and by setting the coefficient of the $dt$ term equal to zero and using \eqref{eqn:endogenous-abm}, we obtain 6 equations for $g_1',g_2',g_3',g_{22}',g_{23}',g_{33}'$ in \eqref{eqn:endogenous-ODEs}. Similarly, we substitute $ \theta_{N,t}$ in \eqref{opt_theta_1} for $\theta_t$ in \eqref{dVN_expand}, and by setting the coefficient of the $dt$ term equal to zero and using \eqref{eqn:endogenous-abm}, we obtain 6 equations for $f_1',f_2',f_3',f_{22}',f_{23}',f_{33}'$ in \eqref{eqn:endogenous-ODEs}.

}

\section{Equilibrium with an exogenous noise trader}\label{section:exogenous}

In this appendix, we consider an equilibrium model with an exogenous noise trader. Following  G\^arleanu and Pedersen~\cite{GP16JET} and Bouchard et.\,al.~\cite{BFHMK18FS}, the noise trader's stock position is exogenously given by $Y$ in \eqref{pp1}.
\begin{definition}[Radner equilibrium with exogenous noise trading]\label{def:exo-eq}
  Trading strategies $\hat\theta_1,\ldots,\hat\theta_I$ and a continuous semimartingale $S=(S_t)_{t\in[0,1]}$ form a \textit{Radner equilibrium with exogenous noise trading} if there exists a measure $\widehat\bQ$ under which $S$ is a local martingale such that
  \begin{enumerate}
    \item \textit{Strategies are optimal:}  For $j=1,\ldots,I$, we have $\hat\theta_j\in\sA(\widehat\bQ)$ solves \eqref{optproblem} with measure $\widehat\bQ$, where $S$ is the corresponding stock price process.
    \item \textit{Markets clear:} We have
      \begin{align}\label{eqn:clearing1}
 \sum_{j=1}^{I} \hat\theta_{j,t}+Y_t=\Sigma,\quad t\in[0,1].
 \end{align} 
  \end{enumerate}
\end{definition}

The value function is again as in~\eqref{eqn:val-fn-form}-\eqref{eqn:term-cond} but for different coefficient functions compared to Theorem~\ref{thm:radner-endogenous}.  The following theorem establishes the existence of a Radner equilibrium with exogenous noise trading.

\begin{theorem}[Radner existence with exogenous noise trading] \label{thm:radner-exogenous} Let $\Sigma\ge0$, $a,\sigma_D^2>0$, and $\sum_{j=1}^I \theta_{j,0-} + Y_0=\Sigma$. Then, there exists a unique smooth solution to the coupled system of ODEs 
for $t\in[0,1]$:
\begin{align*}
  \begin{split}
    g'_{33}(t)&= 2 a \sigma_{Y'}^2 g_{33}(t)^2-\frac{\beta(t)}{I},
    \quad g_{33}(1)=0,\\
    \beta'(t) &= 2 a \sigma_{Y'}^2 g_{33}(t) \beta(t)-\frac{a \sigma_D^2}{I}(1-t),
    \quad \beta(1)=0,
  \end{split}
\end{align*}
such that a Radner equilibrium with exogenous noise trading exists. The equilibrium stock price process is given by
\begin{align}\label{eqn:eq-Shat}
S_t &:=D_t +\mu(t)+\alpha(t) Y_t+\beta(t) Y'_t,\quad t\in[0,1],
\end{align}
where for $t\in[0,1]$,
\begin{align*}
  \mu(t) &:= -\tfrac{a \sigma_D^2 \Sigma }{I}(1-t),\\
  \alpha(t) &:= \tfrac{a \sigma_D^2}{I}(1-t).
\end{align*}

Furthermore, there exists $\widehat\bQ\in\sM$ such that each investor optimally holds $\hat\theta_j\in\sA(\widehat\bQ)$ with
\begin{equation}\label{eqn:eq-thetas}
  \hat\theta_{j,t}=\frac{\Sigma-Y_t}{I}, \quad t\in[0,1],\ j\in\{1,\ldots,I\}.
\end{equation}
\end{theorem}
\begin{proof}
The proof follows along nearly identical lines as the proof of Theorem~\ref{thm:radner-endogenous}. To be specific, by the same way as in Lemma~\ref{en_core_ode_lemma}, we obtain the following ODE result:
\begin{lemma}\label{ex_core_ode_lemma}
Let $a, \sigma_D^2>0$. Then, the following two-dimensional initial value problem has a unique solution for $t\in [0,\infty)$:
\begin{equation}
\begin{split}\label{ex_core_ode}
z_1'(t) &=\left(\tfrac{a \sigma_D^2}{I} \right) t -2a\sigma_{Y'}^2 z_1(t)z_2(t) ,\quad z_1(0)=0,\\
z_2'(t)&= \tfrac{z_1(t)}{I} - 2a \sigma_{Y'}^2 z_2(t)^2, \quad z_2(0)=0.
\end{split}
\end{equation}
\end{lemma}
Then, we define $\alpha, \beta, \mu, g_{33}, g_{23}, g_{22}, g_{3}, g_{2}, g_{1}$ in terms of $z_1$ and $z_2$ in Lemma \ref{ex_core_ode_lemma}:
\begin{equation}
\begin{split}\label{ex_ode_solutions}
\alpha(t)&:=\tfrac{a \sigma_D^2}{I}(1-t),\\
\beta(t)&:=z_1(1-t),\\
\mu(t)&:=-\tfrac{a \sigma_D^2 \Sigma }{I}(1-t),\\
g_{33}(t)&:=z_2(1-t),\\
g_{23}(t)&:=\tfrac{z_1(1-t)}{I},\\
g_{22}(t)&:=\tfrac{a \sigma_D^2}{2I^2}(1-t),\\
g_{3}(t)&:=-\tfrac{ \Sigma }{I}\,  z_1(1-t),\\
g_{2}(t)&:=-\tfrac{a \sigma_D^2 \Sigma }{I^2}(1-t),\\
g_1(t)&:=\tfrac{a \sigma_D^2 \Sigma^2}{2I^2}(1-t) + \sigma_{Y'}^2 \int_t^1 g_{33}(s)\, ds.
\end{split}
\end{equation}
By explicit computations using \eqref{ex_core_ode}, we check that the following ODEs are satisfied:
\begin{align}\label{eqn:exogenous-ODEs}
\begin{split}
g'_1(t) &= \frac{a I^2 \sigma_{Y'}^2 g_{3}(t)^2-a \Sigma ^2 \left(\sigma_D^2+\sigma_{Y'}^2 \beta(t)^2\right)-2 I^2 \sigma_{Y'}^2 g_{33}(t)}{2 I^2},\quad g_1(1)=0,\\
g'_{2}(t)&=  \frac{a \left(I^2 \sigma_{Y'}^2 g_{23}(t) g_{3}(t)+\Sigma  \left(\sigma_D^2+\sigma_{Y'}^2 \beta(t)^2\right)\right)}{I^2},\quad g_{2}(1)=0,\\
g'_3(t)&= 2 a \sigma_{Y'}^2 g_{3}(t) g_{33}(t)-g_{2}(t),\quad g_3(1)=0,\\
g'_{22}(t)&=  -\frac{a \left(\sigma_{Y'}^2 \left(\beta(t)^2-I^2 g_{23}(t)^2\right)+\sigma_D^2\right)}{2 I^2}
,\quad g_{22}(1)=0,\\
g'_{23}(t)&= 2 a \sigma_{Y'}^2 g_{23}(t) g_{33}(t)-2 g_{22}(t),\quad g_{23}(1)=0,\\
g'_{33}(t)&= 2 a \sigma_{Y'}^2 g_{33}(t)^2-g_{23}(t),\quad g_{33}(1)=0,
\end{split}
\end{align}
and
\begin{align}\label{eqn:exogenous-abm}
\begin{split}
\alpha'(t)&=-\frac{a \left(\sigma_{Y'}^2 \beta(t) (\beta(t)-I g_{23}(t))+\sigma_D^2\right)}{I},\quad \alpha(1)=0,\\ 
\beta'(t) &= 2 a \sigma_{Y'}^2 g_{33}(t) \beta(t)-\alpha(t),\quad \beta(1)=0,\\
\mu'(t) &= \frac{a \left(\sigma_{Y'}^2 \beta(t) (I g_{3}(t)+\Sigma  \beta(t))+\Sigma  \sigma_D^2\right)}{I},\quad \mu(1)=0,
\end{split}
\end{align}

Finally, the verification of the admissibility and optimality of \eqref{eqn:eq-thetas} can be checked by the same way as in the proof of Theorem~\ref{thm:radner-endogenous}. The market clearing condition is satisfied also.
\end{proof}

\bibliographystyle{plain}
\bibliography{finance_bib}

\end{document}